\documentclass{article}

\usepackage[square,numbers]{natbib}
 \usepackage[letterpaper, margin=1in]{geometry}
\PassOptionsToPackage{numbers, compress, sort}{natbib}

\usepackage{enumitem}

\usepackage[utf8]{inputenc} 
\usepackage[T1]{fontenc}    
\usepackage{hyperref}       
\usepackage{url}            
\usepackage{booktabs}       
\usepackage{amsfonts}       
\usepackage{nicefrac}       
\usepackage{microtype}      
\usepackage{enumitem}
\usepackage{bbm}
\usepackage[final]{graphicx}
\usepackage{amsmath}
\usepackage{amssymb}
\usepackage{multirow}
\usepackage{mathrsfs}
\usepackage[dvipsnames]{xcolor}

\usepackage{amsthm}

\usepackage{algorithm}
\usepackage[noend]{algorithmic}
\usepackage{mathtools, nccmath}
\usepackage{eqparbox}
\usepackage{comment}
\usepackage{caption}
\usepackage{subcaption}
\usepackage{cleveref}

\newcount\Includeappendix 
\usepackage{tikz}
\usetikzlibrary{automata, arrows, positioning, decorations.pathreplacing, calligraphy}
\usepackage{pgfplots}
\usepackage{pgfplotstable}
\pgfplotsset{compat=1.18}
\usepgfplotslibrary{groupplots}

\usepackage{wrapfig}

\newcount\Comments  
\usepackage{xspace}

\makeatletter
\newcommand{\tpmod}[1]{{\@displayfalse\pmod{#1}}}
\makeatother

\newtheorem{theorem}{Theorem}
\newtheorem{lemma}{Lemma}

\newtheorem{definition}{Definition}
\newtheorem{proposition}{Proposition}

\newtheorem{observation}{Observation}

\theoremstyle{definition}
\newtheorem{example}{Example}

\newenvironment{proofof}[1]{\begin{proof}[\textnormal{\textbf{Proof of \Cref{#1}}}]}{\end{proof}}

\newcommand{\ovm}{\overline{m}}
\newcommand{\unm}{\underline{m}}

\DeclareMathOperator*{\argmax}{arg\,max}

\newcommand{\ol}{\gamma}

\title{
Data Sharing with a Generative AI Competitor
}

\author{
Boaz Taitler%
\thanks{Equal contribution.}
\thanks{%
    {Technion---Israel Institute of Technology (\url{boaztaitler@campus.technion.ac.il})}}
\and Omer Madmon \footnotemark[1]
\thanks{%
    {Technion---Israel Institute of Technology (\url{omermadmon@campus.technion.ac.il})}}
\and Moshe Tennenholtz%
\thanks{%
    {Technion---Israel Institute of Technology (\url{moshet@technion.ac.il})}}
\and Omer Ben{-}Porat%
\thanks{%
    {Technion---Israel Institute of Technology (\url{omerbp@technion.ac.il})}}
}

\begin{document}

\maketitle
\begin{abstract}
    As GenAI platforms grow, their dependence on content from competing providers, combined with access to alternative data sources, creates new challenges for data-sharing decisions.
    In this paper, we provide a model of data sharing between a content creation firm and a GenAI platform that can also acquire content from third-party experts. The interaction is modeled as a Stackelberg game: the firm first decides how much of its proprietary dataset to share with GenAI, and GenAI subsequently determines how much additional data to acquire from external experts. Their utilities depend on user traffic, monetary transfers, and the cost of acquiring additional data from external experts. We characterize the unique subgame perfect equilibrium of the game and uncover a surprising phenomenon: The firm may be willing to pay GenAI to share the firm's own data, leading to a costly data-sharing equilibrium. We further characterize the set of Pareto improving data prices, and show that such improvements occur only when the firm pays to share data. Finally, we study how the price can be set to optimize different design objectives, such as promoting firm data sharing, expert data acquisition, or a balance of both. Our results shed light on the economic forces shaping data-sharing partnerships in the age of GenAI, and provide guidance for platforms, regulators and policymakers seeking to design effective data exchange mechanisms.
\end{abstract}

\section{Introduction}

Generative AI (GenAI) has revolutionized our lives by delivering accurate, personalized content based on user requests within seconds. 
GenAI-driven tools can solve complex tasks that would otherwise take significantly longer, such as generating text and images based on specific user queries, enhancing productivity and streamlining creative processes across various domains.
 This capability is largely due to the vast amounts of data used in training these systems \cite{openai_instruction_following}.

 The fact that GenAI-powered platforms heavily rely on high-volume and high-quality datasets creates a market for data, enabling traditional content creation firms to share their original content with GenAI platforms. 
For instance, OpenAI has launched data partnership programs aimed at collaborating with organizations to enrich its training datasets \cite{openai_data_partnerships}. Similarly, Google has expanded its AI partner ecosystem to enhance GenAI capabilities \cite{google_cloud_ai_partners}. Another notable case is OpenAI’s partnership with Reddit, where OpenAI gains access to Reddit’s Data API to improve user engagement through its products \cite{openai_reddit_partnership}. These partnerships highlight the evolving relationship between GenAI platforms and data holders, which can range from competition to collaboration.

Our work is therefore motivated by the following concrete use case. A content creation firm generates original articles, videos, and multimedia content.
This firm can share its data with a GenAI platform that also competes with it in content distribution and audience reach. Data sharing often happens through a financial transaction, where the firm licenses its content to the GenAI platform for a fee or when the GenAI platform might pay the firm for exclusive access to its high-quality data to enhance its generative models.
Interestingly, the question of who should pay and who should receive payment in this arrangement is not always straightforward.\footnote{Due to the non-trivial question of who the paying entity is, we refer to 'data sharing' as the process of data exchange, with the paying entity clarified by context.}
On one hand, the content creation firm may demand compensation for providing valuable training data. On the other hand, On the other hand, it might even consider paying for the opportunity to share its content with the GenAI platform, recognizing that doing so could reduce the platform's need or capacity to acquire additional data from external experts, whereas allowing the platform to rely solely on external data might lead to a less favorable outcome for the firm.

In these interactions, the GenAI platform often faces a strategic choice between two methods of acquiring new data: either purchasing data from content providers (who may also be competitors), or hiring external experts to generate fresh data (or label existing data) at a cost. Both options have trade-offs: Buying data from a competitor can be inexpensive, but the competitor may be unwilling to sell, whereas generating data through experts is costly but preserves greater independence. Understanding how these two sources of data interact and how the availability of expert-generated data influences the incentives of content providers to share their proprietary datasets is central to the dynamics.

Consequently, the content creation firm also faces a strategic decision regarding data sharing. Sharing its proprietary content with the GenAI platform can bring financial benefits but may also strengthen a direct competitor. This creates a dilemma: should the firm view data sharing as an opportunity or a risk? The decision is further influenced by the GenAI platform's response, as it can substitute the firm's data with expert-generated content if necessary.

The interaction between GenAI platforms and competing content providers raises several important economic questions. What motivates a content platform or creator to share its proprietary data with a GenAI platform, especially when doing so may strengthen a competitor? Under what conditions would a content provider not only agree to share data but even be willing to pay to do so? Conversely, how should a GenAI platform structure incentives to encourage data sharing while minimizing reliance on costly expert data acquisition?

Answering these questions requires careful consideration of the incentives shaped by data pricing. A central aspect of this interaction is the price per unit of data exchanged between the firm and the GenAI platform. In practice, the price could be set by the GenAI platform, by the data provider, or influenced by external factors such as market forces or regulation. Rather than modeling the price-setting process explicitly, we treat the price as an exogenous parameter and analyze how different price levels affect equilibrium outcomes. This approach allows us to understand the strategic behavior of both parties across a range of possible pricing scenarios. We then build on this analysis to study how a designer concerned with specific objectives, such as maximizing data sharing, expert data acquisition, or overall content quality, might optimally select the price.

\paragraph{Our Contribution}
In this paper, we address these questions by modeling the interaction between a content creation firm and a GenAI competitor as a two-stage Stackelberg game, which we call the \emph{data-sharing game}. Our contributions are threefold. First, we develop a game-theoretic model where the firm, first decides the extent of data sharing, followed by the GenAI decision on acquiring additional expert data. This framework captures the strategic interplay between the two entities.
Second, we characterize the subgame perfect equilibrium (SPE) of the game, providing insights into the equilibrium strategies of both the firm and the GenAI platform under various pricing scenarios. Our characterization reveals that under mild assumptions on the prices of both shared data and expert data, two types of equilibria arise:
\begin{theorem}[informal version of Theorem \ref{thrm:eq-analysis}]
    In the data-sharing game between the firm and GenAI, the unique SPE has one of the two forms:
    \begin{itemize}[noitemsep, topsep=-4pt, leftmargin=*]
        \item The firm shares the amount of data that makes GenAI indifferent between buying all expert data and not buying expert data at all; GenAI does not buy expert data.
        \item The firm shares the amount of data that maximizes its utility under the assumption that GenAI is forced to complete its dataset by buying all expert data; GenAI buys all expert data.
    \end{itemize}
\end{theorem}

We further demonstrate the robustness of the equilibria and their impact on the utilities of the firm and GenAI through sensitivity analysis. Third, we use our equilibrium analysis to draw several important results regarding the economic implications of data sharing.

\begin{itemize}[noitemsep, topsep=-4pt, leftmargin=*]
        \item \textbf{Costly Data Sharing (Proposition \ref{cor:costly-ds}).} We uncover a surprising phenomenon where the firm may be willing to \emph{pay} GenAI to share the firm's own data, leading to costly data-sharing equilibria.
        \item \textbf{Pareto-Improving Data Sharing (Proposition \ref{cor:pareto}).} We identify conditions under which data-sharing agreements can be Pareto improving, benefiting both parties involved.
        \item \textbf{Optimal Data Pricing (Proposition \ref{cor:welfare}).} We analyze how the price of data influences equilibrium outcomes and characterize optimal pricing rules under different design objectives, such as promoting firm data sharing or expert data acquisition.
    \end{itemize}

Our analysis sheds light on the economic forces shaping data-sharing partnerships in the era of generative AI and provides actionable insights for platforms, content creators, and policymakers in designing effective data exchange mechanisms.

\paragraph{Related Work}

This work contributes to the growing literature on the strategic and societal aspects of foundation models and machine learning \cite{conitzerposition, laufer2024fine,shaham2025privacy,hammond2025multi}. This line of research includes studies motivated by social choice theory \cite{fish2023generative}, mechanism design \cite{duetting2024mechanism}, and welfare concerns in competitive environments, both when GenAI systems are used \cite{taitler2025braess, taitler2025selective, raghavan2024competition} and in competition against them \cite{yao2024human_p4, esmaeili2024strategize}.
Other related works examine the broader societal effects of machine learning in applications such as recommendation systems \cite{hron2022modeling, jagadeesan2206supply, yao2024rethinking, iclr_paper_6}. This includes research focused on the design of recommendation algorithms \cite{ben2018game, yao2024user_p5} as well as works inspired by information retrieval settings \cite{madmon2023search, madmon2025on}. As highlighted by \citet{rosenfeld2025machine} and \citet{dean2024accounting}, mathematical modeling plays a crucial role in helping planners and decision-makers understand the societal impacts of AI technologies and design systems that better promote social welfare.

Our work is closely related to the literature on data sharing, studying Pareto-improving sharing mechanisms both with and without monetary payments \cite{gradwohl_faact_2022,gradwohl_jair_2023,gafni2024prediction}.
Most closely related are the works of \citet{tsoy2023strategic} and \citet{gradwohl_tark_2023}, both of which model competitive environments involving data exchange. The former studies competition between learning algorithms that improve with access to more data, while the latter considers traditional firms that use data to offer personalized services. Our model can be viewed as a bridge between these two, as it involves both a traditional firm and a learning algorithm. 
More broadly, our work relates to the literature on data acquisition \cite{zhang2024survey, falconer2025selling, hossain2023equilibrium}. This includes studies on eliciting truthful and accurate information \cite{chen2020truthful, faltings2022game}, valuing data sources \cite{xu2024data,galperti2024value}, addressing privacy concerns \cite{sim2023incentives}, and leveraging data to gain advantage in a competitive setting \cite{hardt2022performative, bergemann2024data, jagadeesan2023competition, prufer2021competing} among others.
Our work also contributes to the existing literature on data sharing by introducing an outside option for the agent with whom data is shared, in the form of expert data.

\section{Model}

\paragraph{Overview} 
A content creation firm (\emph{Firm}) and a Generative AI platform (\emph{GenAI}) interact in a Stackelberg (two-stage) game. In the first stage, Firm decides what portion of its dataset to share with GenAI. In the second stage, GenAI determines how much additional data to purchase from external experts. The utilities of both parties depend on the traffic each platform attracts from users, monetary exchanges between the parties, and the cost incurred by GenAI for acquiring expert data.

\paragraph{Notations} 
Let $\alpha \in [0,1]$ denote Firm's decision (the fraction of its dataset shared with GenAI), and $x \in [0,1]$ denote GenAI's action (the amount of data it purchases from external experts, in addition to the data acquired from Firm). The total data volume accessible to GenAI is bounded, implying $x \leq 1 - \alpha$. Let $m \in \mathbb{R}$ represent the per-unit price for data exchanged between Firm and GenAI, which may be positive (Firm sells data to GenAI) or negative (Firm pays GenAI to share its data).\footnote{As we demonstrate next (Section \ref{sec:econ-implications}), by solving the game for any exogenously set $m$, one can identify market prices that satisfy desirable economic properties from different perspectives, such as maximizing the equilibrium payoff of GenAI or Firm, enhancing social welfare, or achieving Pareto improvements over the baseline case of no data sharing.} Additionally, let $c > 0$ denote the cost per data unit paid by GenAI to external experts. A traffic function $T(\alpha, x)$ determines the proportion of users that choose Firm based on the actions of both parties, where the remaining proportion, $1 - T(\alpha, x)$, corresponds to users who choose GenAI. Finally, let $r^f$ and $r^g$ represent the marginal rewards from user engagement for Firm and GenAI, respectively.

\paragraph{Timing and Utilities} 
The interaction unfolds as follows. First, Firm selects $\alpha \in [0,1]$. Next, GenAI observes $\alpha$ and chooses $x \in [0,1-\alpha]$. Firm's utility is:

\begin{equation}
  U(\alpha, x) = T(\alpha, x) r^f + m\alpha,
\end{equation}

and GenAI's utility is:

\begin{equation}
    V(\alpha, x) = \big(1 - T(\alpha, x)\big) r^g - cx - m\alpha.
\end{equation}

\paragraph{Solution Concept} 
As standard in Stackelberg games, we adopt the \emph{subgame perfect equilibrium (SPE)} as the solution concept. In this framework, the leading player (Firm) maximizes its payoff, anticipating that GenAI will respond optimally to any Firm action. This results in the following bi-level optimization problem:

\begin{equation}
    \max_{\alpha \in [0,1]} U\big(\alpha, x^\star(\alpha)\big),
\end{equation}

where $x^\star(\alpha) \coloneqq \argmax_{x \in [0,1-\alpha]} V(\alpha, x)$ represents GenAI's best response to Firm's choice, with ties resolved in favor of the minimal $x$.\footnote{This tie-breaking rule reflects GenAI’s default inclination to avoid purchasing expert data when indifferent, which may stem from a conservative stance on data acquisition.}

\paragraph{Traffic Function}
The specification of traffic functions encapsulates critical assumptions about market dynamics, including user behavior in response to available content and the strategic incentives of both Firm and the GenAI platform. As a result, the traffic function serves as a crucial analytical tool to understand how data sharing and acquisition decisions shape competition and user engagement.
While traffic functions can take various forms, our analysis focuses on the following traffic function:

\[
    T(\alpha, x) = (1 - \alpha)(1 - x).
\]

This formulation captures the intuition that users gravitate toward Firm when GenAI lacks relevant content. Given that the user is interested in a single data point, $x$ represents the probability that GenAI acquired the data point from experts, and $\alpha$ represents the probability that it was acquired from Firm. In this case, the user chooses Firm only if GenAI did not obtain the desired data point, which occurs with probability $(1 - \alpha)(1 - x)$. Notably, the multiplicative structure also reflects the decreasing marginal utility of data for GenAI: the more data Firm shares (higher $\alpha$), the smaller the incremental traffic benefit GenAI gains from acquiring additional expert data (increasing $x$).\footnote{Another underlying assumption captured by the structure of $T$ is the symmetry between the two data sources. This can be easily broken by introducing weighting factors or asymmetries in the traffic function, e.g., replacing $(1 - \alpha)(1 - x)$ with $(1 - w_1 \alpha)(1 - w_2 x)$ for weights $w_1, w_2 \in (0,1]$. However, for ease of exposition and to highlight the core strategic trade-offs, we retain the symmetric formulation in our analysis.}

While our main exposition focuses on the above traffic function, our analysis extends to a broader class of traffic functions that preserve key qualitative properties such as monotonicity and diminishing returns. In particular, the results presented in the appendix generalize to traffic functions in which the extent to which the expert's data serves as a replacement for Firm's data is controlled by an \emph{overlap parameter} $\ol$, where our traffic function boils down to a special case where the overlap is maximal, hence the two data source are substitutes. A complete discussion of this richer family of traffic functions appears in Appendix \ref{app:traffic-fns}.

\paragraph{Model Assumptions}
We now outline and justify two assumptions underlying our model, and discuss the contexts in which they are plausible. 

\begin{itemize}[noitemsep, topsep=-4pt, leftmargin=*]
    \item \textbf{GenAI buys all data offered by Firm.} We assume that when Firm offers a portion $\alpha$ of its dataset, GenAI fully accepts the offer and cannot reject or partially purchase it. This assumption reflects real-world arrangements such as \emph{output contracts}, in which a buyer agrees to purchase the entire output provided by a seller \cite{cornell_output_contract}. Such arrangements are common in data-driven industries.\footnote{For example, in the data annotation sector, such as with platforms like Outlier or DataAnnotation, firms frequently pay annotators on an hourly basis, effectively agreeing to purchase all labeled data generated during that period \cite{dataannotation_homepage,outlier_homepage}. Similarly, AI platforms may commit to ingesting all content produced by contracted partners to ensure predictable data flows.} In our setting, this assumption is particularly natural when the parties have a pre-established collaboration, where the price $m$ and GenAI’s obligation to accept any offer from Firm are contractually fixed. These arrangements streamline interactions and promote predictability.
    
    \item \textbf{Aggregated GenAI dataset volume is bounded.} We assume that GenAI is limited in the total volume of data it can hold, which includes both the data shared by Firm and the data purchased from external experts. One justification for this assumption is the practical constraint that storing and processing large volumes of data incurs significant costs, making it unfeasible for GenAI to handle unlimited amounts of data. Additionally, GenAI may face technological or regulatory limitations on data storage.\footnote{Alternatively, this assumption can be reinterpreted through the traffic function: if $\alpha + x$ exceeds the upper limit of 1, the content quality of GenAI becomes so high relative to Firm's that it captures the entire user base. In such cases, the traffic simplifies to $T \equiv 0$ for Firm and $1-T \equiv 1$ for GenAI, effectively modeling the same outcome without explicitly imposing a hard data volume constraint.}
\end{itemize}

\begin{example}
\label{example-1}
Consider the instance $r^f = r^g = 1$, $c = 0.32$ and $m = -0.1$. We now analyze three possible strategy profiles:
\begin{enumerate}[noitemsep, topsep=-4pt, leftmargin=*]
    \item \textbf{No Data Acquisition:} Firm opts not to share any data, i.e., $\alpha = 0$. Given this choice, GenAI is free to select any $x \in [0, 1]$. When GenAI does not acquire data from either Firm or the experts, it sets $x = 0$. In this scenario, all users select Firm, yielding $T(0, 0) = 1$, with corresponding utilities $U(0, 0) = r^f = 1$ for Firm and $V(0, 0) = 0$ for GenAI.

    \item \textbf{No Data Sharing:} Firm again selects $\alpha = 0$, but GenAI now chooses to purchase data from the experts. Since $c < r^g$, it is strictly beneficial for GenAI to do so. The optimal choice in this case is $x = 1$, which maximizes $V(0, x)$. As a result, all users prefer GenAI, leading to $T(0, 1) = 0$, and the resulting utilities are $U(0, 1) = 0$ for Firm and $V(0, 1) = r^g - c = 0.68$ for GenAI.

    \item \textbf{Equilibrium:} By choosing to share a positive amount of data, Firm can reduce GenAI’s ability to acquire expert data and thereby retain a fraction of the user base. In this example, the optimal strategy for Firm is to set $\alpha = 0.68$, which induces a best response of $x = 0$ from GenAI. The resulting equilibrium is $(\alpha, x) = (0.68, 0)$, yielding utilities $U(0.68, 0) = 0.252$ for Firm and $V(0.68, 0) = 0.748$ for GenAI. In Section \ref{sec:equil}, we characterize the unique equilibrium in any instance.
\end{enumerate}
\end{example}

\section{Equilibrium Characterization}
\label{sec:equil}

We now turn to the analysis of the data-sharing game, focusing on characterizing the equilibrium strategies of both players. This characterization forms the foundation for understanding their strategic behavior and deriving economic insights. In the next section, we will leverage the equilibrium structure to discuss the broader implications of data-sharing dynamics. Before stating our results, we introduce several additional key notations.

As we are interested in SPE, the game is solved using backward induction, meaning we begin by solving for GenAI's best response to any fixed Firm action $\alpha$. Since for any fixed $\alpha$ the utility of GenAI is linear in its action $x$, the best response will be either $x=0$ or $x=1-\alpha$, depending on the slope of this linear utility function. For a given game instance, we denote by $R_G$ the value of $\alpha$ for which GenAI is indifferent between playing $x=0$ and $x=1-\alpha$, and refer to it as the \emph{indifference threshold} of GenAI. A simple calculation reveals that $R_G$ can be written as follows:
\[
R_G = 1 -  \frac{c}{r^{g}}.
\]

Next, we notice that solving Firm optimization problem involves a delicate case analysis, depending on whether its action $\alpha$ falls above or below the indifference threshold of GenAI. In the case where $\alpha$ falls below $R_G$, we observe that Firm's utility is quadratic in $\alpha$, and we denote by $R_F$ its \emph{unrestricted} unique optimum (i.e., on $\mathbb{R}$ rather than $[0,1]$), which can be written as follows:
\[
R_F = \frac{1}{2} + \frac{m}{2 r^{f}}.
\]

We refer to this quota as Firm's \emph{sharing level under forced completion}, as it represents the optimal data-sharing choice that maximizes Firm’s utility under the assumption that GenAI completes its dataset with external expert data.

Our analysis primarily focuses on the case where \( 0 \leq R_G, R_F \leq 1 \). This serves as a \emph{regularity condition}, ensuring that the pricing mechanism \( m \) and the expert's cost \( c \) are within reasonable bounds, thereby preventing degenerate solutions in the equilibrium analysis. However, In Appendix~\ref{app:eq-analysis}, we present a comprehensive equilibrium characterization for a richer class of traffic functions, including cases where the regularity condition do not hold.

Assuming the regularity condition holds, Firm's optimal decision hinges on the relationship between \( R_F \) and \( R_G \). When \( R_F < R_G \), two candidate actions emerge. The first is Firm’s optimal choice within the interval \([0, R_G)\), which is \( R_F \) by definition. The second is \( R_G \), the optimal action in the interval \([R_G, 1]\), where Firm’s utility becomes linear with a negative slope due to GenAI’s best response (to refrain from purchasing expert data). Selecting the optimal action thus reduces to comparing Firm’s utilities under GenAI’s respective responses: \( U(R_G, 0) \) and \( U(R_F, 1 - R_F) \). Figure \ref{fig:rf_below_rg} illustrates Firm’s utility function (given GenAI’s best response), highlighting the two possible subcases: \( U(R_G, 0) \le U(R_F, 1 - R_F) \) (shown in \textcolor{green}{green} rectangle) and \( U(R_G, 0) > U(R_F, 1 - R_F) \) (shown in \textcolor{red}{red} triangle).

In the remaining case, where \( R_F \ge R_G \), we show that it is always optimal for Firm to choose \( R_G \). Unlike the previous case, the quadratic region of Firm's utility function cannot exceed $\alpha = R_G$, making \( R_G \) the unique optimal choice.\footnote{
This follows from the fact that $U(R_G, 1 - R_G) \le U(R_G, 0)$ and that $U(\alpha, 0)$ is decreasing in $\alpha \in [R_G, 1]$.
}
This case is depicted in Figure \ref{fig:rf_above_rg}.  
We are now ready to formally state our equilibrium characterization result:

\begin{figure}[t]
    \centering

    \begin{subfigure}{0.45\textwidth}
        \centering
        \includegraphics[width=\linewidth]{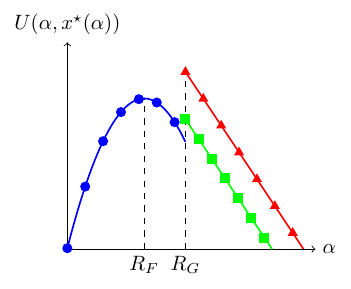}
        \caption{$R_F < R_G$}
        \label{fig:rf_below_rg}
    \end{subfigure}
    \hfill
    \begin{subfigure}{0.45\textwidth}
        \centering
        \includegraphics[width=\linewidth]{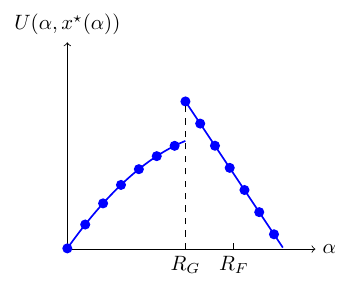}
        \caption{$R_F \ge R_G$}
        \label{fig:rf_above_rg}
    \end{subfigure}

    \caption{The utility of Firm (given the best reply of GenAI) as a function of its data sharing level $\alpha$.}
    \label{fig:equilibrium_analysis}
\end{figure}

\begin{theorem}
\label{thrm:eq-analysis}
    Assume that the regularity condition \( 0 \leq R_G, R_F \leq 1 \) holds. Then, the unique SPE\footnote{To be more precise, the SPE is $(\alpha^{eq}, x^\star)$, and $(\alpha^{eq}, x^{eq})$ is the on-path choice of the players under the SPE.} in the data-sharing game is given by:
    {\small
    \begin{align*}
    (\alpha^{eq}, x^{eq}) = \begin{cases}
    \left(R_G, 0 \right) & \mbox{$R_G > R_F$ and $U \left(R_G, 0 \right) > U \left(R_F, 1-R_F \right)$ or $R_F \geq R_G$} \\
    \left(R_F, 1-R_F \right) & \mbox{$R_G > R_F$ and $U \left(R_G, 0 \right) \leq U \left(R_F, 1-R_F \right)$}
    \end{cases}
    \end{align*}}
\end{theorem}

The proof of Theorem \ref{thrm:eq-analysis} follows immediately from the complete equilibrium analysis provided in Appendix \ref{app:eq-analysis}, which also covers the cases in which the regularity condition does not hold.

\section{Economic Implications of Data Sharing}
\label{sec:econ-implications}

In this section, we utilize the equilibrium characterization of the game to derive several economic insights on data-sharing markets (with proofs being deferred to Appendix \ref{app:omitted-proofs}).

\paragraph{Costly Data Sharing} Our equilibrium analysis shows that Firm might be willing to \emph{pay} to share its data with GenAI, leading to a \emph{costly data-sharing equilibrium}. In particular, it enables us to provide sufficient conditions on the price $m$ and the experts' cost $c$ for such a phenomenon to occur:

\begin{proposition}
\label{cor:costly-ds}
If $c \in \left( 0,r^g \right)$ and $m \in (-r^f, 0)$, the unique SPE $(\alpha^{eq}, x^{eq})$ satisfies $\alpha^{eq} > 0$.
\end{proposition}

The existence of instances that admit a costly data-sharing equilibrium hints that firms may perceive data sharing not merely as an expense but as a strategic investment. 
Recall from Example~\ref{example-1} that Firm can always guarantee a utility of at least zero by choosing not to share any data. Therefore, whenever Firm pays GenAI to share its data, the decision must be utility-improving and can thus be interpreted as a form of investment.

\paragraph{Pareto-Improving Data Sharing} 
An important question is whether data-sharing agreements can be structured to benefit both parties relative to the baseline of no sharing. In particular, we seek to identify conditions under which a suitable choice of the data price $m$ leads to a \emph{Pareto improvement}: both Firm and GenAI achieve (weakly) higher payoffs than if Firm refrains from sharing its data. Pareto-improving data sharing is economically significant, as it ensures voluntary participation from both sides and supports the stability of data-sharing arrangements without external enforcement.

\begin{definition}
    A data price $m$ is said to be \textbf{Pareto improving} if $U(\alpha^{eq}, x^{eq}) \ge U(0, x^\star(0))$ and $V(\alpha^{eq}, x^{eq}) \ge V(0, x^\star(0))$.
\end{definition}

The following proposition characterizes the existence and structure of Pareto-improving prices:
    
\begin{proposition} \label{cor:pareto}
    Assume that the regularity condition holds. There exists $M = [m_1, m_2] \subset \mathbb{R}_{\leq 0}$ such that:
    \begin{enumerate}[noitemsep, topsep=-4pt, leftmargin=*]
        \item If $R_G > 0$ and $r^g \geq 2r^f$ then $m$ is Pareto improving if and only if $m \in M$.
        \item If $R_G > 0$ and $r^g < 2r^f$, there exist $M' = [m'_1, m'_2] \subset \mathbb{R}_{\leq 0}$ such that $M \cap M' = \emptyset$ and $m$ is Pareto improving if and only if $m \in M \cup M'$.
        \item If $R_G = 0$, there are no Pareto improving prices $m$.
    \end{enumerate}
\end{proposition}

\begin{wrapfigure}{r}{0.5\textwidth}
    \centering
    \includegraphics{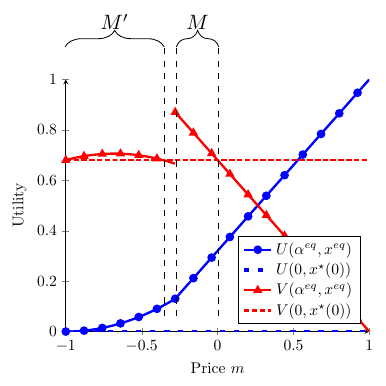}
    \caption{An illustration of the set of Pareto improving prices corresponding to Example \ref{example-1}. Notice that this example falls into the second case of Proposition \ref{cor:pareto}, as the set of prices splits into two disjoint (non-positive) intervals. }
    \label{fig: pareto}
\end{wrapfigure}

Figure \ref{fig: pareto} illustrates the set of Pareto improving prices in Example \ref{example-1}.
Interestingly, Pareto improvements occur only under \emph{negative mechanisms}, meaning that Firm \textbf{pays} GenAI to share its data. Intuitively, GenAI buys all the data from experts when $r^g > c$, and Firm refrains from sharing data. In this case, $U(0, 1) = 1$ and $V(0, 1) = r^g - c$. Therefore, acquiring data from Firm is detrimental to GenAI's user base, and as a result, accepting $\alpha > 0$ is profitable only when $m \leq 0$. We further extend this proposition in Appendix~\ref{app:omitted-proofs} to a broader class of traffic functions (introduced in Appendix~\ref{app:traffic-fns}), showing that \emph{positive mechanisms}, in which Firm \textbf{receives} payment for its data, can also lead to Pareto-improving outcomes.

\paragraph{Optimal Data Pricing} 

Thus far, we have treated the per-unit data price $m$ as an exogenous parameter. However, in practice, it is natural to ask who sets the price and according to what objective. For example, if the GenAI platform sets the price, it might aim to maximize Firm's willingness to share data, as having access to such a dataset can benefit the long-term interests of GenAI, potentially even beyond the scope of our data-sharing game.

In contrast, a regulator seeking to promote overall content creation and data availability may prefer a balanced outcome, encouraging both Firm data sharing and GenAI's expert data acquisition. 
Other scenarios are also conceivable: for instance, an external policymaker concerned primarily with the quality of GenAI's content might prioritize expert data acquisition over Firm data sharing.\footnote{This perspective can also be framed as a social welfare optimization problem from the users' standpoint. For instance, if expert data is of higher quality than Firm data, user welfare could be modeled as a weighted sum of data quantities acquired by GenAI from each source, where the weights reflect relative data quality.}

To capture these different objectives, we consider linear combinations of Firm's sharing level $\alpha$ and GenAI's expert data acquisition $x$. Specifically, we introduce a parameter $\lambda \geq 0$ that governs the relative weight placed on expert data acquisition, and study the following optimization problem:

\begin{align}
\label{eq: welfare optimization}
    & \max_{m \in [-r^f, r^f]} \alpha + \lambda x \tag{$P_{\lambda}$} \nonumber \\
    & \text{s.t. $(\alpha, x)$ is the SPE outcome in the game with price $m$} \nonumber
\end{align}

Setting $\lambda = 0$ corresponds to valuing only Firm's data sharing, while larger values of $\lambda$ place increasing emphasis on GenAI's expert data acquisition. In particular, letting $\lambda \to \infty$ captures the case where the price setter is primarily concerned with maximizing expert data purchases. The case of $\lambda = 1$ corresponds to valuing Firm's data sharing and GenAI's expert data acquisition equally.
The following proposition characterizes the optimal data prices for the objective defined in~\eqref{eq: welfare optimization}:

\begin{proposition} \label{cor:welfare}
    Assume that the regularity condition holds. 
    If $r^g < 2c$, then any $m \in [-r^f, r^f]$ is an optimal solution for ~\eqref{eq: welfare optimization}. Otherwise, any solution of ~\eqref{eq: welfare optimization} is of the form
    {\small
    \begin{align*}
        m = \begin{cases}
        -r^f & \mbox{$\lambda \geq 1$} \\
        m^b & \mbox{$\lambda \in (0.5,1)$} \\
        m' & \mbox{$\lambda = 0.5$} \\
        m'' & \mbox{$\lambda < 0.5$} \\
    \end{cases}
    \end{align*}}%

    where $m^b = r^f\left(4R_G - 3\right)$ and $m', m''$ can be any prices in $[m^b, r^f]$ such that $m'' > m^b$.
\end{proposition}

The price $m^b$ plays a key role in \Cref{cor:welfare}, as this value, according to \Cref{thrm:eq-analysis}, marks the boundary between the two equilibria. Specifically, the profile $(R_F, 1 - R_F)$ is the equilibrium whenever $m \in [-r^f, m^b]$, while the profile $(R_G, 0)$ is the equilibrium whenever $m \in (m^b, r^f]$. Furthermore, if $r^g < 2c$ then $m^b < -r^f$ and the only equilibrium is $(R_G, 0)$, which is not affected by $m$. Otherwise, the pricing $m$ controls which equilibrium GenAI and Firm would end up in. Larger values of $\lambda$ are associated with the equilibrium $(R_F, 1-R_F)$ as low prices reduces $\alpha = R_F$. The prices $m = -r^f$ and $m = r^f(4R_G - 3)$ induces the minimal and maximal values of $R_F$ in the equilibrium profile $(R_F, 1-R_F)$. Observe that $\alpha \leq R_G$ and therefore when $\lambda$ is small, $\alpha$ is maximized and the optimal profile which maximizes \Cref{thrm:eq-analysis} shifts to $(R_G, 0)$, which is the equilibrium when $m \in \left(r^f \left(4R_G - 3 \right), r^f \right]$. Lastly, $\lambda = 0.5$ is the tipping point where both equilibria can maximize our objective.

\section{Sensitivity Analysis}

\newlength{\subfigwidth}
\setlength{\subfigwidth}{0.5\textwidth}

In this section, we illustrate the utilities of Firm and GenAI across different instances. To this end, we conduct a sensitivity analysis to examine the boundaries and robustness of each equilibrium with respect to the parameters $r^f, r^g, c, m$.

Our first analysis explores varying values of $r^f$ and $r^g$, while keeping the pricing fixed at $c = m = 1$, as shown in Figure~\ref{fig_sim_rf_rg}. Light colors in Figure~\ref{fig3a} and Figure~\ref{fig3b} indicate high utility values, while darker colors indicate lower utilities. Figure~\ref{fig3c} illustrates the equilibrium for each instance: yellow represents the profile $(R_G, 0)$, and purple represents the profile $(R_F, 1 - R_F)$. More elaborately, in Figure~\ref{fig3a}, higher values of $r^f$ make it more profitable for Firm to generate revenue from user traffic. In the $(R_G, 0)$ equilibrium, GenAI does not purchase data from experts, so all its data originates from Firm. Observe that as $r^g$ decreases, the equilibrium value $R_G$ also decreases, leading Firm to sell less data and retain a larger user base.
As $r^g$ increases, the equilibrium shifts to $(R_F, 1 - R_F)$, which is independent of $r^g$ and therefore does not affect Firm's utility.

In Figure~\ref{fig3b}, we observe that GenAI's utility remains constant with respect to $r^f$ under the $(R_G, 0)$ equilibrium. As $r^g$ increases, GenAI generates higher revenue from user traffic. We highlight that in both Figure~\ref{fig3a} and Figure~\ref{fig3b}, the boundary between the two equilibria induces a discontinuity in the utilities of both Firm and GenAI, as anticipated by Figure~\ref{fig:equilibrium_analysis}.

\begin{figure}
    \raisebox{2.6cm}{
            $r^g$
        }
    \centering
    \begin{subfigure}{0.25\textwidth}
    \centering
        \includegraphics[width=\textwidth, trim=20 0 0 0, clip]{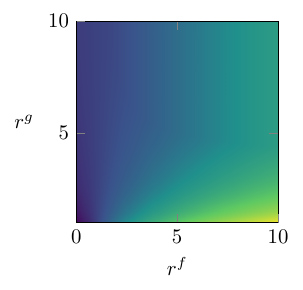}
    \caption{Firm's utility} \label{fig3a}
    \end{subfigure}
    \hfill
    \begin{subfigure}{0.25\textwidth}
    \centering
        \includegraphics[width=\textwidth, trim=20 0 0 0, clip]{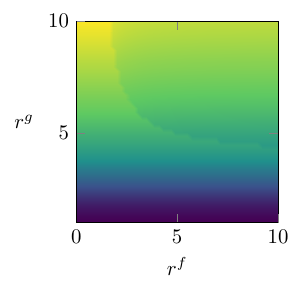}
    \caption{GenAI's utility} \label{fig3b}
    \end{subfigure}
    \hfill
    \begin{subfigure}{0.25\textwidth}
    \centering
        \includegraphics[width=\textwidth, trim=20 0 0 0, clip]{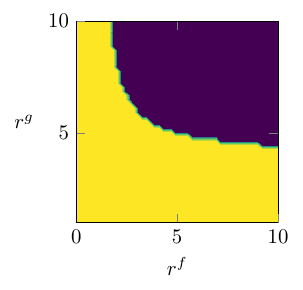}
    \caption{Equilibrium} \label{fig3c}
    \end{subfigure}

    \vspace{1em}

    \raisebox{2.6cm}{
            $m$
        }
    \centering
    \begin{subfigure}{0.25\textwidth}
    \centering
        \includegraphics[width=\textwidth, trim=20 0 0 0, clip]{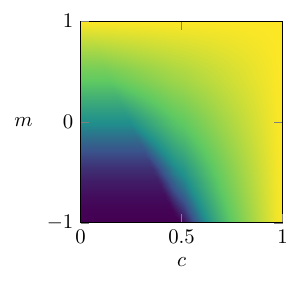}
    \caption{Firm's utility} \label{fig4a}
    \end{subfigure}
    \hfill
    \begin{subfigure}{0.25\textwidth}
    \centering
        \includegraphics[width=\textwidth, trim=20 0 0 0, clip]{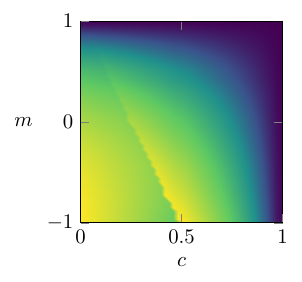}
    \caption{GenAI's utility} \label{fig4b}
    \end{subfigure}
    \hfill
    \begin{subfigure}{0.25\textwidth}
    \centering
        \includegraphics[width=\textwidth, trim=20 0 0 0, clip]{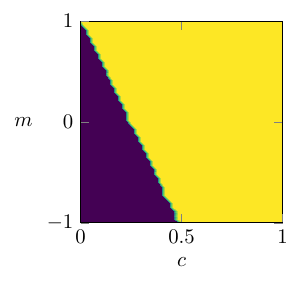}
    \caption{Equilibrium} \label{fig4c}
    \end{subfigure}
    
\caption{Sensitivity analysis for $r^f, r^g, c, m$. The top row (figures a–c) varies $r^f$ and $r^g$: figures a and b show the utilities of the Firm and GenAI, respectively, while figure c presents the induced equilibrium for each parameter combination. The bottom row (figures d–f) varies $c$ and $m$: figures d and e describe the corresponding utilities, and figure f shows the resulting equilibrium.}
\label{fig_sim_rf_rg}
\end{figure}

Our next analysis focuses on varying $c$ and $m$, while fixing $r^f = r^g = 1$. We observe a similar interplay between the instance parameters and the utilities in each equilibrium profile. Recall from \Cref{cor:welfare} that the boundary between the two equilibria is given by $m = r^f(4R_G - 3) = r^f\left(1 - 4 \nicefrac{c}{r^g}\right)$, which scales linearly with $c$ and aligns with Figure~\ref{fig4c}. Since Firm is a utility maximizer, the lowest utility in Figure~\ref{fig4a} (indicated by dark colors) is bounded below by 0. Notably, we observe the costly data sharing phenomena for which Firm's utility increases when $m$ is negative, while in the equilibrium profile $(R_F, 1-R_F)$.

\section{Conclusions and Future Work}

In this work, we analyzed the economic dynamics underlying data-sharing agreements between content creation firms and GenAI platforms. Using a game-theoretic model, we characterized the subgame perfect equilibrium of a two-stage game where Firm first decides how much of its proprietary dataset to share, and the GenAI platform subsequently chooses how much additional data to acquire from external experts. Our analysis uncovered several novel insights, including the emergence of costly data-sharing equilibria in which firms are willing to pay GenAI platforms to share their data, and the identification of conditions under which data-sharing agreements can be Pareto improving. We also studied how different objectives can be optimized through data pricing, providing guidelines for platforms and policymakers aiming to promote efficient and mutually beneficial data exchanges.

Several future directions emerge naturally from our model. First, while we focused on an exogenously set price for data sharing, an important extension would be to endogenize the price-setting process through bargaining or mechanism design. Second, our model considers a single firm and a single GenAI platform; extending the framework to competitive environments involving multiple firms or multiple GenAI platforms could reveal richer strategic behavior and market-level effects. Third, our baseline traffic function assumes symmetric substitutability between firm and expert data; exploring alternative functional forms that capture asymmetric quality or reputation differences could provide further practical insights. Finally, introducing incomplete information about the traffic function or about the other player's utility function presents a promising direction for future research. Such models capture scenarios where firms and GenAI platforms have only partial knowledge of user behavior or their competitor’s incentives, and must make decisions under uncertainty.

\bibliographystyle{plainnat}
\bibliography{main}

\appendix

\section{Additional Traffic Functions}
\label{app:traffic-fns}

In this section, we introduce and discuss a family of traffic functions that extends and generalizes the one discussed in the main paper. Crucially, our main result in Appendix \ref{app:eq-analysis} (the SPE characterization) is derived with respect to this general family. As a result, the economic insights analogous to those discussed in the main paper can be easily derived using the same analytical framework. Our family of traffic functions includes traffic functions of the form

\[
    T(\alpha, x) = 1 - \alpha - x + \ol\alpha x,
\]

where \( \ol \in (0,1] \) is an \emph{overlap parameter} that quantifies the extent to which GenAI's quality is affected by redundant data sources.
Notably, the case of $\ol = 1$ (i.e., perfect overlap) boils down to the traffic function presented in the main paper.\footnote{In the case where $\ol=0$ the bi-level optimization problem is linear for both agents, meaning that there always exists a trivial extreme point solution. In contrast, for $\ol>0$ this holds only for GenAI’s best response, making the analysis of the game more complex, as explored in the paper.}

This structure can be interpreted as follows.
The term $\alpha+x$ represents the \emph{accuracy} in GenAI's response, which increases in its available data volume. Naturally, the traffic to Firm decreases as GenAI becomes more accurate, hence more appealing to end-users.
The additional term $\ol \alpha x$ represents the \emph{datasets overlap}.
When \( \ol = 0 \), the quality of GenAI's dataset remains independent of overlap, meaning that data duplication has no impact. Conversely, as \( \ol \) increases, overlap between Firm’s data and externally acquired data diminishes GenAI’s dataset quality. This reduction in quality increases user preference for Firm's content, thereby driving more traffic toward it.

Importantly, these traffic functions retain the key structural properties discussed in the main paper: monotonicity in both arguments, diminishing marginal influence of each data source, and data source symmetry.
Moreover, for any \( \ol \), the traffic function is bi-linear, ruling out the possibility of multiple equilibria in the data-sharing game.

\section{Complete Equilibrium Analysis}
\label{app:eq-analysis}

We repeat the definition of the traffic function
\begin{align*}
    T = 1-\alpha-x + k\alpha x,
\end{align*}
where the utility of GenAI is defined by
\begin{align*}
V = (1-T)r^g - cx - m\alpha,
\end{align*}
and the utility of Firm is defined by
\begin{align*}
U = T r^f + m\alpha.
\end{align*}

Notice that when clear from context, we write $T, V, U$ instead of $T(\alpha,x), V(\alpha,x), U(\alpha,x)$.
In our game, Firm takes the first action and then GenAI chooses. Therefore, when GenAI chooses its action $x$, it already knows the action of Firm $\alpha$. Consequently, we solve this game in 2 steps: In the first step, we start by solving GenAI's optimization problem given that $\alpha$ is fixed. In the second step, we solve Firm's optimization problem given that GenAI played the best-response $x^\star(\alpha)$ for every given action of Firm.

\textbf{GenAI's optimization problem} is given by:

\begin{align*}
\max_x V(\alpha, x) &= \max_x (1-T) r^g - cx - m\alpha \\
&=\max_x (1-(1-\alpha-x + \ol \alpha x)) r^g - cx - m\alpha \\
&= \max_x (x-\ol \alpha x) r^g - cx \\
&= \max_x x\left( (1 - \ol \alpha) r^g - c \right)
\end{align*}

Thus, $x$ is chosen depending on the slop $(1 - \ol \alpha) r^g - c$. Formally, it is given by
\begin{align} \label{eq: opt x nonvisibility}
x = \begin{cases}
    1-\alpha & \mbox{$(1 - \ol \alpha) r^g > c$} \\
    0 & \mbox{Otherwise}
\end{cases}
\end{align}

Notice that the condition in Equation~\eqref{eq: opt x nonvisibility} is equal to
\begin{align*}
    \alpha < \frac{r^g - c}{\ol  r^g}
\end{align*}
Therefore, Equation~\eqref{eq: opt x nonvisibility} can be rewritten as

\begin{align*}
x = \begin{cases}
    1-\alpha & \mbox{$\alpha < \frac{r^g - c}{\ol  r^g}$} \\
    0 & \mbox{Otherwise}
\end{cases}
\end{align*}

We note that if $\alpha = \frac{r^g - c}{\ol r^g}$ then GenAI is indifferent between $x = 1-\alpha$ and $x = 0$.
Therefore, from our indifference assumption, we assume that if $\alpha = \frac{r^g - c}{\ol r^g}$ then $x = 0$.
That is, for general $\ol$, the indifference threshold of GenAI can be written as
\[
R_G = \frac{r^g - c}{\ol  r^g}.
\]

Now, the \textbf{Firm's optimization problem} is:

\begin{align*}
\max_{\alpha} U(\alpha, x) &= \max_{\alpha} T r^f + m\alpha \\
&= \max_{\alpha} (1-\alpha - x + \ol \alpha x) r^f + m\alpha
\end{align*}

We now separate the analysis into 2 cases.

\begin{itemize}
\item $\alpha \geq \frac{r^g - c}{\ol  r^g}$ \quad: in this case $x = 0$ and therefore

\begin{align*}
\alpha_1 = \argmax_{\alpha} -\alpha r^f + m\alpha = \argmax_{\alpha} \alpha(m - r^f)
\end{align*}
Thus, $\alpha$ is chosen according to the slop $m-r^f$, formally:
\begin{align*}
\alpha_1 = \begin{cases}
    1 & \mbox{$m \geq r^f$} \\
    \max \left\{0, \frac{r^g - c}{\ol  r^g}\right\} & \mbox{Otherwise}.
\end{cases}
\end{align*}

\item $\alpha < \frac{r^g - c}{\ol  r^g}$ \quad: in this case $x = 1-\alpha$ and therefore
\begin{align*}
& \max_{\alpha} (1-\alpha - (1-\alpha) + \ol \alpha (1-\alpha))) r^f + m\alpha \\
& \max_{\alpha} \ol \alpha (1-\alpha) r^f + m\alpha \\ 
& \max_{\alpha} (\ol r^f + m) \alpha - \ol  r^f \alpha^2
\end{align*}

To find the optimal $\alpha$ we take the derivative:
\begin{align*}
0 = \ol r^f + m - 2\ol  r^f \alpha
\end{align*}

Therefore, the optimal $\alpha$ in this case is $\alpha = \frac{\ol r^f + m}{2\ol  r^f}$. Put differently, for general $\ol$, Firm’s sharing level under forced completion can be written as
\[
R_F = \frac{\ol r^f + m}{2\ol  r^f}.
\]
Together with the condition $\alpha \leq \frac{r^g - c}{\ol  r^g}$, we get the following solution:

\begin{align*}
\alpha_2 = \max \left\{0, \min \left\{ \frac{r^g - c}{\ol  r^g} - \varepsilon, \frac{\ol r^f + m}{2\ol  r^f}, 1\right\} \right\}
\end{align*}

such that $\lim \varepsilon \rightarrow 0^+$.
Whenever $\frac{r^g - c}{\ol r^g} \in [0, 1]$ then we need to compare $U(\alpha_1, 0)$ with $U(\alpha_2, 0)$. 
Observe that if $\alpha_2 = \frac{r^g - c}{\ol r^g} - \varepsilon$ for $\varepsilon \rightarrow 0^+$, then it holds that
\begin{align*}
    U(\frac{r^g - c}{\ol r^g} - \varepsilon, 1-\frac{r^g - c}{\ol r^g} + \varepsilon) \rightarrow U(\frac{r^g - c}{\ol r^g}, 1-\frac{r^g - c}{\ol r^g}) < U(\frac{r^g - c}{\ol r^g}, 0) \leq U(\alpha_1, 0)
\end{align*}

\end{itemize}

The overall solution considering both cases is

\begin{align*}
\alpha = \begin{cases}
    \alpha_1 & \mbox{$U(\alpha_1, 0) \geq U(\alpha_2, 1-\alpha_2)$} \\
    \alpha_2 & \mbox{Otherwise}
\end{cases}
\end{align*}

\paragraph{Cases analysis} We can now explore it further depending on the parameters of the instance.
\begin{itemize}
\item If $c > r^g$ then it holds that $\alpha_2$ is not defined and therefore the equilibrium is $(\alpha, x) = (\alpha_1, 0)$, where
\begin{align*}
\alpha_1 = \begin{cases}
    1 & \mbox{$m \geq r^f$} \\
    0 & \mbox{Otherwise}
\end{cases}.
\end{align*}
\end{itemize}

From here on, we analyze different cases in which $c \leq r^g$.

\begin{itemize}
\item If $\frac{r^g - c}{\ol r^g} \geq 1$ and $\frac{\ol r^f + m}{2\ol r^f} \geq 1$: This corresponds to the case where $m \geq \ol r^f$, for which $\alpha_1 = \alpha_2 = 1$ and therefore $x = 0$. Thus the optimal solution is
\begin{align*}
    (\alpha, x) = (1, 0).
\end{align*}

\item If $\frac{r^g - c}{\ol r^g} \geq 1 > \frac{\ol r^f + m}{2\ol r^f} \geq 0$: then $\alpha_1 = 1$ and $\alpha_2 = \frac{\ol r^f + m}{2\ol r^f}$. This case is only relevant if $m \in [-\ol r^f, \ol r^f)$. The optimal solution is given by
\begin{align*}
    \alpha = \begin{cases}
        \alpha_1 & \mbox{$U(1, 0) > U(\frac{\ol r^f + m}{2\ol r^f}, \frac{\ol r^f - m}{2\ol r^f})$} \\
        \alpha_2 & \mbox{Otherwise}.
    \end{cases}
\end{align*}
Observe that $U(1, 0) = m$ and 
\begin{align*}
U \left( \frac{\ol r^f + m}{2\ol r^f}, \frac{\ol r^f - m}{2\ol r^f} \right) &= U(\alpha_2, 1-\alpha_2) \\
        &= (1-\alpha_2 - (1-\alpha_2) + \ol  \alpha_2 (1-\alpha_2))r^f + \alpha_2 m \\
        &= \ol  \alpha_2 (1-\alpha_2) r^f + \alpha_2 m \\
        &= \alpha_2 \left(\ol r^f (1-\alpha_2) + m \right) \\
        &= \frac{\ol r^f + m}{2\ol r^f} \left(\ol r^f \frac{2\ol r^f - \ol r^f - m}{2\ol r^f} + m \right) \\
        &= \frac{\ol r^f + m}{2\ol r^f} \left(\frac{\ol r^f - m}{2} + m \right) \\
        &= \frac{\ol r^f + m}{2\ol r^f} \frac{\ol r^f + m}{2} \\
        &= \frac{(\ol r^f + m)^2}{4\ol r^f}.
\end{align*}
Therefore, the conditions is equivalent to $m > \frac{(\ol r^f + m)^2}{4\ol r^f}$. We use the following lemma (whose proof appears at the end of this section):

\begin{lemma} \label{lemma: condition1}
For every $m \in [-y, y]$ it holds that $\frac{(y + m)^2}{4y} \geq m$.
\end{lemma}
Thus, we are left with only one solution:
\begin{align*}
(\alpha, x) = \left(\frac{\ol r^f + m}{2\ol r^f}, \frac{\ol r^f - m}{2\ol r^f} \right).
\end{align*}

\item If $\frac{r^g - c}{\ol r^g} \geq 1 > 0 > \frac{\ol r^f + m}{2\ol r^f}$: then $\alpha_1 = 1$ and $\alpha_2 = 0$. This case is relevant only for $m < -\ol r^f$, in which the optimal $\alpha$ is given by
\begin{align*}
\alpha = \begin{cases}
    1 & \mbox{$U(1, 0) > U(0, 1)$} \\
    0 & \mbox{Otherwise}
\end{cases}
\end{align*}
Notice that $U(0, 1) = 0 > -\ol r^f > m = U(1, 0)$. Thus, there is only one solution in this case
\begin{align*}
(\alpha, x) = (0, 1).
\end{align*}

\item If $1 > \frac{r^g - c}{\ol r^g} > \frac{\ol r^f + m}{2\ol r^f} \geq 0$: This case is relevant only for $m \in [-\ol r^f, \ol r^f)$, for which $\alpha_1 = \frac{r^g - c}{\ol r^g}$ and $\alpha_2 = \frac{\ol r^f + m}{2\ol r^f}$. Therefore the optimal $\alpha$ is given by
\begin{align*}
    \alpha = \begin{cases}
        \alpha_1 & \mbox{$U(\frac{r^g - c}{\ol r^g}, 0) > U(\frac{\ol r^f + m}{2\ol r^f}, 1-\frac{\ol r^f + m}{2\ol r^f})$} \\
        \alpha_2 & \mbox{Otherwise}
    \end{cases}.
\end{align*}
Observe that $U(\frac{r^g - c}{\ol r^g}, 0) = r^f + \frac{r^g - c}{\ol r^g}(m-r^f)$ and therefore we get that
\begin{align*}
    (\alpha, x) = \begin{cases}
        (\frac{r^g - c}{\ol r^g}, 0) & \mbox{$r^f + \frac{r^g - c}{\ol r^g}(m-r^f) > \frac{(\ol r^f + m)^2}{4\ol r^f}$} \\
        (\frac{\ol r^f + m}{2\ol r^f}, \frac{\ol r^f - m}{2\ol r^f}) & \mbox{Otherwise}
    \end{cases}.
\end{align*}

\item If $1 > \frac{r^g - c}{\ol r^g} > 0 > \frac{\ol r^f + m}{2\ol r^f}$: In this case, it holds that $m < -\ol r^f$. Therefore, $\alpha_1 = \frac{r^g - c}{\ol r^g}$ and $\alpha_2 = 0$. The optimal $\alpha$ is given by
\begin{align*}
        \alpha = \begin{cases}
            \alpha_1 & \mbox{$U(\frac{r^g - c}{\ol r^g}, 0) > U(0, 1)$} \\
            \alpha_2 & \mbox{Otherwise}
        \end{cases}.
\end{align*}
By rewriting the condition condition, we get that
\begin{align*}
        (\alpha, x) = \begin{cases}
            (\frac{r^g - c}{\ol r^g}, 0) & \mbox{$r^f + \frac{r^g - c}{\ol r^g}(m-r^f) > 0$} \\
            (0, 1) & \mbox{Otherwise}
        \end{cases}.
\end{align*}

\item If $\frac{\ol r^f + m}{2\ol r^f} \geq 1 > \frac{r^g - c}{\ol r^g} \geq 0$: This case is relevant only if $m \geq \ol r^f$. Furthermore, From our discussion, $\alpha_2$ is suboptimal in this and therefore, We can further split this case into 2 sub-cases:
\begin{enumerate}
    \item $m \in [\ol r^f, r^f)$: for which $\alpha_1 = \frac{r^g - c}{\ol r^g}$ and the solution is given by
    \begin{align*}
    (\alpha, x) = \left( \frac{r^g - c}{\ol r^g}, 0 \right).
    \end{align*}
    
    \item $m \geq r^f$: for which $\alpha_1 = 1$ and the solution is
    \begin{align*}
        (\alpha, x) = (1, 0).
    \end{align*}
    
\end{enumerate}

\item If $1 > \frac{\ol r^f + m}{2\ol r^f} > \frac{r^g - c}{\ol r^g} \geq 0$: This case is relevant only if $m \in (-\ol r^f, \ol r^f)$, in which case $\alpha_1 = \frac{r^g - c}{\ol r^g}$ and $\alpha_2$ is suboptimal. Thus, we conclude that
\begin{align*}
(\alpha, x) = \left(\frac{r^g - c}{\ol r^g}, 0 \right)
\end{align*}
\end{itemize}

We can now summarize:
\begin{align*}
(\alpha, x) = \begin{cases}
    (1, 0) & \mbox{$c > r^g$ and $m \geq r^f$} \\
           & \mbox{or $\frac{r^g - c}{\ol r^g} \geq 1$ and $\frac{\ol r^f + m}{2\ol r^f} \geq 1$} \\
           & \mbox{or $\frac{\ol r^f + m}{2\ol r^f} \geq 1 > \frac{r^g - c}{\ol r^g} \geq 0$ and $m \geq r^f$} \\
    \\
    (0, 0) & \mbox{$c > r^g$ and $m < r^f$}  \\
    &  \\
    (0, 1) & \mbox{$\frac{r^g - c}{\ol r^g} \geq 1 > 0 > \frac{\ol r^f + m}{2\ol r^f}$} \\ 
           & \mbox{or $1 > \frac{r^g - c}{\ol r^g} > 0 > \frac{\ol r^f + m}{2\ol r^f}$ and $r^f + \frac{r^g - c}{\ol r^g}(m-r^f) \leq 0$} \\
    & \\
    \left(\frac{r^g - c}{\ol r^g}, 0 \right) & \mbox{$1 > \frac{r^g - c}{\ol r^g} \geq \frac{\ol r^f + m}{2\ol r^f} \geq 0$ and $r^f + \frac{r^g - c}{\ol r^g}(m-r^f) > \frac{(\ol r^f + m)^2}{4\ol r^f}$} \\
    & \mbox{or $1 > \frac{r^g - c}{\ol r^g} > 0 > \frac{\ol r^f + m}{2\ol r^f}$ and $r^f + \frac{r^g - c}{\ol r^g}(m-r^f) > 0$} \\
    & \mbox{or $\frac{\ol r^f + m}{2\ol r^f} \geq 1 > \frac{r^g - c}{\ol r^g} \geq 0$ and $m \in [\ol r^f, r^f)$} \\
    & \mbox{or $1 > \frac{\ol r^f + m}{2\ol r^f} > \frac{r^g - c}{\ol r^g} \geq 0$} \\
    & \\
    \left(\frac{\ol r^f + m}{2\ol r^f}, \frac{\ol r^f - m}{2\ol r^f} \right) & \mbox{$\frac{r^g - c}{\ol r^g} \geq 1 > \frac{\ol r^f + m}{2\ol r^f} \geq 0$} \\
           & \mbox{or $1 > \frac{r^g - c}{\ol r^g} > \frac{\ol r^f + m}{2\ol r^f} \geq 0$ and $r^f + \frac{r^g - c}{\ol r^g}(m-r^f) \leq \frac{(\ol r^f + m)^2}{4\ol r^f}$} \\
    & 
\end{cases}.
\end{align*}

Further simplifying the conditions yields:
\begin{align*}
(\alpha, x) = \begin{cases}
    (1, 0) & \mbox{$m \geq r^f$} \\
           & \mbox{or $\frac{r^g - c}{\ol r^g} \geq 1$ and $m \geq \ol r^f$} \\
    \\
    (0, 0) & \mbox{$c > r^g$ and $m < r^f$}  \\
    &  \\
    (0, 1) & \mbox{$\frac{r^g - c}{\ol r^g} \geq 1$ and $m < -\ol r^f$} \\ 
           & \mbox{or $1 > \frac{r^g - c}{\ol r^g} \geq 0$ and $m < -\ol r^f$ and $r^f + \frac{r^g - c}{\ol r^g}(m-r^f) \leq 0$} \\
    & \\
    \left(\frac{r^g - c}{\ol r^g}, 0 \right) & \mbox{$1 > \frac{r^g - c}{\ol r^g} > \frac{\ol r^f + m}{2\ol r^f} \geq 0$ and $r^f + \frac{r^g - c}{\ol r^g}(m-r^f) > \frac{(\ol r^f + m)^2}{4\ol r^f}$} \\
    & \mbox{or $1 > \frac{r^g - c}{\ol r^g} > 0 > \frac{\ol r^f + m}{2\ol r^f}$ and $r^f + \frac{r^g - c}{\ol r^g}(m-r^f) > 0$} \\
    & \mbox{or $1 > \frac{r^g - c}{\ol r^g} \geq 0$ and $m \in [\ol r^f, r^f)$} \\
    & \mbox{or $1 > \frac{\ol r^f + m}{2\ol r^f} \geq \frac{r^g - c}{\ol r^g} \geq 0$} \\
    & \\
    \left(\frac{\ol r^f + m}{2\ol r^f}, \frac{\ol r^f - m}{2\ol r^f} \right) & \mbox{$\frac{r^g - c}{\ol r^g} \geq 1 > \frac{\ol r^f + m}{2\ol r^f} \geq 0$} \\
           & \mbox{or $1 > \frac{r^g - c}{\ol r^g} > \frac{\ol r^f + m}{2\ol r^f} \geq 0$ and $r^f + \frac{r^g - c}{\ol r^g}(m-r^f) \leq \frac{(\ol r^f + m)^2}{4\ol r^f}$} \\
    & 
\end{cases}.
\end{align*}

Plugging in the notations of $R_G$ and $R_F$:

\begin{align*}
(\alpha, x) = \begin{cases}
    (1, 0) & \mbox{$m \geq r^f$} \\
           & \mbox{or $R_G \geq 1$ and $m \geq \ol r^f$} \\
    \\
    (0, 0) & \mbox{$c > r^g$ and $m < r^f$}  \\
    &  \\
    (0, 1) & \mbox{$R_G \geq 1$ and $m < -\ol r^f$} \\ 
           & \mbox{or $1 > R_G \geq 0$ and $m < -\ol r^f$ and $r^f + \frac{r^g - c}{\ol r^g}(m-r^f) \leq 0$} \\
    & \\
    \left(R_G, 0 \right) & \mbox{$1 > R_G > R_F \geq 0$ and $r^f + \frac{r^g - c}{\ol r^g}(m-r^f) > \frac{(\ol r^f + m)^2}{4\ol r^f}$} \\
    & \mbox{or $1 > R_G > 0 > R_F $ and $r^f + \frac{r^g - c}{\ol r^g}(m-r^f) > 0$} \\
    & \mbox{or $1 > R_G \geq 0$ and $m \in [\ol r^f, r^f)$} \\
    & \mbox{or $1 > R_F \geq R_G \geq 0$} \\
    & \\
    \left(R_F, 1-R_F \right) & \mbox{$R_G \geq 1 > R_F \geq 0$} \\
           & \mbox{or $1 > R_G > R_F \geq 0$ and $r^f + \frac{r^g - c}{\ol r^g}(m-r^f) \leq \frac{(\ol r^f + m)^2}{4\ol r^f}$} \\
    & 
\end{cases}.
\end{align*}

Finally, restricting to the cases in which the regularity condition $0 \le R_G, R_F \le 1$ holds, we get exactly the result stated as Theorem \ref{thrm:eq-analysis}:

\begin{align*}
(\alpha, x) = \begin{cases}
    \left(R_G, 0 \right) & \mbox{$1 \geq R_G > R_F \geq 0$ and $U \left(R_G, 0 \right) > U \left(R_F, 1-R_F \right)$} \\
    & \mbox{or $1 \geq R_F \geq R_G \geq 0$} \\
    & \\
    \left(R_F, 1-R_F \right) & \mbox{or $1 \geq R_G > R_F \geq 0$ and $U \left(R_G, 0 \right) \leq U \left(R_F, 1-R_F \right)$}
\end{cases}.
\end{align*}

Notice that in particular, we obtain that Theorem \ref{thrm:eq-analysis} holds for general $\ol$.

This concludes the proof of \Cref{thrm:eq-analysis}.

\begin{proofof}{lemma: condition1}
Denote $h(m) = \frac{(y + m)^2}{4y}$ and $g(m) = m$. Our goal is to show that for every $m \in [-y, y]$, it holds that 
\begin{align}
    h(m) \geq g(m) \label{eq: h geq g}.
\end{align} 
Observe that
\begin{align*}
    h(y) = \frac{(2y)^2}{4y} = y = g(y)
\end{align*}
Therefore, for $m = y$ the inequality~\eqref{eq: h geq g} holds with an equality. Next, we take the derivative.
\begin{align*}
\frac{dh(m)}{dm} = \frac{2(y+m)}{4y} = \frac{y+m}{2y} = \frac{1}{2} + \frac{m}{2y}.
\end{align*}
Notice that for every $m \in [-y, y)$ it holds that $\frac{m}{2y} < \frac{1}{2}$ and therefore $\frac{dh(m)}{dm} < 1$.
On the other hand, it holds that $\frac{dg(m)}{dm} = 1$. Thus, we get that if we start at $m = y$, then $h(m) = g(m)$, but as we decrease the value of $m$, we get that $g(m)$ decreases faster than $h(m)$ for every $m \in [-y, y)$. Therefore, $h(m) > g(m)$ in this range.
This concludes the proof of Lemma~\ref{lemma: condition1}.
\end{proofof}

\section{Omitted Proofs (Section \ref{sec:econ-implications})}
\label{app:omitted-proofs}

\subsection{Proof of Proposition \ref{cor:costly-ds}}

We present and prove the proposition for general $\ol \in (0,1]$, hence the proof of $\ol=1$ (the version that appears in the main paper) follows immediately.

\begin{proposition}
    (Extension of Proposition \ref{cor:costly-ds} for general $\ol$) If $c \in \left( (1-\ol)r^g,r^g \right)$, then for every $m \in (-\ol r^f, 0)$ the unique SPE $(\alpha^{eq}, x^{eq})$ satisfies $\alpha^{eq} > 0$.
\end{proposition}

\begin{proof}
    The conditions on $c$ and $m$ imply that the regularity condition holds with $R_F, R_G > 0$, hence by Theorem \ref{thrm:eq-analysis} (which holds for general $\ol$ by Appendix \ref{app:eq-analysis}) we have that $\alpha^{eq}$ equals either $R_F$ or $R_G$, with both being strictly positive.
\end{proof}

\subsection{Proof of \Cref{cor:pareto}}

We present and prove the proposition for general $\ol \in (0,1]$. The proof for $\gamma = 1$ follows afterwards.

\begin{proposition} \label{cor: pareto with gamma}
    (Extension of Proposition \ref{cor:pareto} for general $\ol$) Assume that the regularity condition holds. Let

    \begin{align*}
        \ovm = \begin{cases}
            \min \left\{ \gamma r^f, r^g(1-\gamma), \frac{r^f (\gamma r^g + 2c)}{r^g + 2r^f} , r^f \frac{\gamma r^g - 2c}{r^g - 2r^f}\right\} & \mbox{$r^g < 2r^f$} \\
            \min \left\{ \gamma r^f, r^g(1-\gamma), \frac{r^f (\gamma r^g + 2c)}{r^g + 2r^f} \right\} & \mbox{$r^g > 2r^f$}
        \end{cases}    
    \end{align*}
    
    and 
    \begin{align*}
        \unm = \begin{cases}
            \max \left\{ - \gamma r^f, - \frac{1-R_G}{R_G}r^f\right\} & \mbox{$r^g < 2r^f$} \\
            \max \left\{ - \gamma r^f, - \frac{1-R_G}{R_G}r^f, r^f \frac{\gamma r^g - 2c}{r^g - 2r^f}\right\} & \mbox{$r^g > 2r^f$} 
        \end{cases}
    \end{align*}
    
    Then,
    \begin{enumerate}
        \item If $R_G > 0$, every price $m \in [\unm, \ovm]$ is Pareto improving compared to forcing no data sharing (i.e., $\alpha = 0$).
        \item If $R_G = 0$, there are no Pareto improving prices $m$.
    \end{enumerate}
\end{proposition}

\begin{proof}
Without data sharing, we fix $\alpha = 0$. In this case, the optimal action of GenAI is $x = 1$ if $R_G > 0$ or $x = 0$ if $R_G = 0$.

\paragraph{Case I. $R_G > 0$} Our base line is $V(0, 1) = r^g - c > 0$ and $U(0, 1) = 0$. As noted, there are two possible solutions for this game, which Firm has to choose from.

\begin{itemize}
    \item If $(\alpha^{eq}, x^{eq}) = (R_G, 0)$ then:
    \begin{enumerate}
        \item GenAI's perspective: $V(R_G, 0) = R_G r^g - R_G m \geq V(0, 1)$ and therefore we get that 
        \begin{align}
            m \leq r^g(1-\gamma). \label{eq: rg equalibrium upper}
        \end{align}
        \item Firm's perspective: $U(R_G, 0) = (1-R_G)r^f + R_G m \geq U(0, 1)$ and therefore we get that 
        \begin{align}
            m \geq -\frac{1-R_G}{R_G}r^f. \label{eq: rg equalibrium lower}
        \end{align}.
    \end{enumerate}

    \item If $(\alpha^{eq}, x^{eq}) = (R_F, 1-R_F)$ then:
    \begin{enumerate}
        \item GenAI's perspective: $V(R_F, 1-R_F) = (1-\gamma R_F(1-R_F)) r^g - c(1-R_F) - m R_F \geq V(0, 1)$
        We get the following inequality
        \[
            -\gamma R_F(1-R_F) r^g + R_F c - R_F m \geq 0, 
        \]
        Thus, $m$ has to satisfy one of the following conditions:
        \begin{itemize}
            \item $r^g > 2r^f$: 
            \begin{align}
                m \geq r^f \frac{\gamma r^g - 2c}{r^g - 2r^f} \label{eq: rg above 2rf}
            \end{align}
            \item $r^g < 2r^f$:
            \begin{align}
                m \leq r^f \frac{\gamma r^g - 2c}{r^g - 2r^f} \label{eq: rg below 2rf}
            \end{align}
        \end{itemize}
    
        \item Firm's perspective: $U(R_F, 1-R_F) = \gamma R_F (1-R_F)r^f + mR_F \geq U(0, 1)$.
        We get the following inequality
        \[
            \gamma (1-R_F)r^f + m \geq 0
        \]
        By extracting $m$, we get that $m \geq -\gamma r^f$.
    \end{enumerate}
\end{itemize}

\paragraph{Case II. $R_G = 0$} in this case, the base line is $(\alpha, x) = (0, 0)$. Therefore, it holds that $V(0, 0) = 0$ and $U(0, 0) = r^f$. Furthermore, we consider only the instances where $R_F \geq R_G$ and therefore the equilibrium is $(\alpha^{eq}, x^{eq}) = (R_G, 0)$

 \begin{enumerate}
     \item GenAI's perspective: $V(R_G, 0) = R_G r^g - R_G m \geq V(0, 0)$ and therefore we get that $m \leq r^g$.
     \item Firm's perspective: $U(R_G, 0) = (1-R_G)r^f + R_G m \geq U(0, 0)$ and therefore we get that $m \geq r^f$.
     Notice that from the condition $0 \leq R_F \leq 1$, it holds that $m \in [-\gamma r^f, \gamma r^f]$, and therefore, there is no $m$ under our conditions that induces a Pareto-improving equilibrium.
 \end{enumerate}

 This concludes the proof of \Cref{cor: pareto with gamma}.
 \end{proof}

We finished with the proof for any $\gamma \in (0, 1]$ and now turn to analyze for $\gamma = 1$.

We begin by presenting the following lemma.
\begin{lemma} \label{lemma: equalibrium RF conditions}
    Let $\gamma = 1$. The profile $(R_F, 1-R_F)$ is an equilibrium if and only if
    \begin{align*}
        m \leq r^f(4R_G - 3).
    \end{align*}
\end{lemma}

\Cref{lemma: equalibrium RF conditions} implies that $m = r^f(4R_G - 3)$ is the boundary for the two equilibria. That is, the profile $(R_F, 1-R_F)$ is the equilibrium for $m \leq r^f(4R_G - 3)$, and $(R_G, 0)$ is the equilibrium when $m > r^f(4R_G - 3)$. Therefore, the values of $m$ that are Pareto-improving can be separated into $2$ disjoint ranges, each induced by a different equilibrium.

We separate our analysis for two cases, namely when $r^g > 2r^f$ and $r^g \leq 2r^f$.

\begin{itemize}
    \item If $r^g > 2r^f$: 
We analyze this case for each equilibrium. Observe that if $(\alpha^{eq}, x^{eq}) = (R_F, 1-R_F)$ then according to \Cref{lemma: equalibrium RF conditions} it must hold that $m \leq r^f(4R_G - 3)$. Therefore, we can now check whether $m \leq r^f(4R^G - 3)$ also satisfies Equation~\eqref{eq: rg above 2rf}.

\begin{align*}
   m \leq r^f(4R_G - 3) = r^f\frac{r^g - 4c}{r^g} < r^f\frac{r^g - 4c}{r^g - 2r^f} < r^f\frac{r^g - 2c}{r^g - 2r^f}.
\end{align*}
Therefore, if $r^g > 2r^f$ then there is no $m \leq r^f(4R_G - 3)$ that results in Pareto improvement. In other words, there are no Pareto-improving prices $m$ that induce the equilibrium $(R_F, 1-R_F)$.
We now turn to the equilibrium profile $(R_G, 0)$. According to \Cref{cor: pareto with gamma} and Equation~\eqref{eq: rg equalibrium lower}, the equilibrium profile $(R_G, 0)$ can be induced by Pareto-improving prices $m$ if and only if $m$ lies in
\[
M_{R_G} = \left( \max\left\{-r^f, -r^f \frac{1-R_G}{R_G}, r^f(4R_G - 3)\right\}, 0 \right].
\]
Thus, to conclude, in the case where $r^g > 2r^f$, the Pareto-improving prices $m \in M_{R_G}$ for which the induced equilibrium is $(R_G, 0)$.

    \item If $r^g < 2r^f$:

In this case, we have 2 possible ranges for $m$, one for the equilibrium $(R_F, 1-R_F)$ and another for $(R_G, 0)$.
Starting with $(R_F, 1-R_F)$, then according to Equation~\eqref{eq: rg below 2rf} and \Cref{lemma: equalibrium RF conditions} we get that $m$ is Pareto-improving if it lies in
\begin{align}
    M_{R_F} =  \left[-r_f, \min \left\{ r^f \frac{r^g - 2c}{r^g - 2r^f}, r^f(4R_G - 3) \right\} \right].
\end{align}

Moving on to the equilibrium $(R_G, 0)$, the conditions for Pareto improvement do not depend on whether $r^g < 2r^f$. Therefore, for the same arguments as before, the Pareto-improving prices are $m \in M_{R_G}$.

Thus, to conclude, for the case where $r^g < 2r^f$, any $m \in M_{R_G} \cup M_{R_F}$ is Pareto-improving.

\end{itemize}

We finish by showing that there are no Pareto-improving prices such that $m > 0$. Observe that if $m$ is Pareto-improving, it must either be in $M_{R_F}$ or $M_{R_G}$. By definition it holds that $M_{R_G} \in \mathbb{R}_{\leq 0}$. Next, we use the following observation to show that $M_{R_F}$ is also in $\mathbb{R}_{\leq 0}$.

\begin{observation} \label{obs: negative expressions Mrf}
If $r^g < 2r^f$ then there is no $r^g > 0$ such that $r^f \frac{r^g - 2c}{r^g - 2r^f} > 0$ and $r^f(4R_G - 3) > 0$.
\end{observation}

Thus, it holds that $M_{R_G} \cup M_{R_F} \in \mathbb{R}_{\leq 0}$.

This concludes the proof of \Cref{cor:pareto}.

\begin{proofof}{lemma: equalibrium RF conditions}
There are two conditions that must hold:
\begin{align}
    R_G > R_F   \label{eq: welfare RG above RF}
\end{align}
and
\begin{align}
    U(R_G, 0) \leq U(R_F, 1-R_F)    \label{eq: welfare RF opt}
\end{align}

The condition in Inequality~\eqref{eq: welfare RG above RF} yields
\begin{align*}
    m < \frac{r^f}{r^g} \left( r^g(2-\gamma) - 2c \right)
\end{align*}
while the second conditions from Inequality~\eqref{eq: welfare RF opt} results in
\begin{align*}
    m^2 \frac{1}{4\gamma r^f} + m \left( \frac{1}{2} - R_G \right) + r^f \left( \frac{\gamma}{4} - 1 + R_G \right) \geq 0
\end{align*}
Where the solution of the inequality is given by:
\begin{align*}
m_{\pm} = \frac{R_G - \frac{1}{2} \pm \sqrt{(R_G - 1)\left(R_G - \frac{1}{\gamma} \right)}}{\frac{1}{2r^f}}
\end{align*}

For $\gamma = 1$ we get that:
\begin{align*}
m_{\pm} = 2r^f \left( R_G - \frac{1}{2} \pm \left( 1-R_G \right) \right)
\end{align*}
Thus,
\begin{align*}
    & m_+ = r^f \\
    & m_- = r^f(4R_G - 3).
\end{align*}
Finally, we can conclude the values of $m$ that satisfies Inequality~\eqref{eq: welfare RG above RF} and Inequality~\ref{eq: welfare RF opt} are:
\begin{enumerate}
    \item $m < \frac{r^f}{r^g} \left( r^g(2-\gamma) - 2c \right)$
    \item $m \geq m_+$ or $m \leq m_-$.
\end{enumerate}

Observe that
\begin{align*}
    m_- < r^f\left(1-2\frac{c}{r^g}\right) < r^f
\end{align*}
Therefore, we can summarize that the profile $(R_F, 1-R_F)$ is the equilibrium only if $m \leq m_-$.
This concludes the proof of \Cref{lemma: equalibrium RF conditions}.
\end{proofof}

\begin{proofof}{obs: negative expressions Mrf}
Assume in contradiction that there exists $r^g > 0$ such that $r^f \frac{r^g - 2c}{r^g - 2r^f} > 0$ and $r^f(4R_G - 3) > 0$.

Starting with $r^f \frac{r^g - 2c}{r^g - 2r^f} > 0$, Since $r^f < 2r^f$ it must hold that $r^g < 2c$.

Next, observe that
\begin{align*}
r^f(4R_G - 3) = \frac{r^f}{r^g} \left( r^g - 4c \right) < \frac{r^f}{r^g} \left( 2c - 4c \right) < -2c \frac{r^f}{r^g} < 0.
\end{align*}

Which is a contradiction to our assumption.
\end{proofof}

\subsection{Proof of \Cref{cor:welfare}}

\begin{proof}

The equilibrium analysis reveals two possible equilibria. Let $\theta^{eq} = (\alpha^{eq}, x^{eq})$ denote the equilibrium profile. Therefore, our objective is of the following form:

\begin{align*}
    \alpha + \lambda x &= \begin{cases}
        R_G & \mbox{$\theta^{eq} = (R_G, 0)$} \\
        R_F + \lambda(1-R_F) & \mbox{$\theta^{eq} = (R_F, 1-R_F)$}
    \end{cases} \\
    &= \begin{cases}
        R_G & \mbox{$\theta^{eq} = (R_G, 0)$} \\
        R_F(1-\lambda) + \lambda & \mbox{$\theta^{eq} = (R_F, 1-R_F)$}
    \end{cases}.
\end{align*}

According to \Cref{lemma: equalibrium RF conditions}, the condition for $\theta^{eq} = (R_F, 1-R_F)$ to be an equilibrium is $m \leq m_- = r^f(4R_G - 3)$. 

Notice that if the profile $(R_F, 1-R_F)$ is the equilibrium, then we need to choose $m$ such that
\begin{align*}
    \max_{m} R_F(1-\lambda) + \lambda = \max_{m} R_F(1-\lambda)
\end{align*}
Therefore, we would like to pick the maximal $m$ if $\lambda < 1$ and the minimal $m$ if $\lambda > 1$.

We split our analysis into 4 cases: $\lambda > 1$, $\lambda = 1$, $0 <\lambda < 1$ and $\lambda = 0$.
\begin{enumerate}
    \item Case 1: $\lambda \geq 1$. From our conditions on $m$ and our condition for $R_F \in [0, 1]$ we get that $m_{min}$ is the minimal value in $[-r^f, m_-]$. $m_{min}$ is well defined only if $m_- \geq -r^f$, which implies $r^g \geq 2c$. In this case, the objective function under the profile $(R_F, 1-R_F)$ is given by
        \begin{align*}
            R_F(1-\lambda) + \lambda = \lambda > 1 > R_G.
        \end{align*}
        Since $R_F(1-\lambda) + \lambda > R_G$ we indeed get that $m = -r^f$ maximize our objective function.
        Therefore, to conclude, $m = m_{min}$ is optimal if $\lambda \geq 1$ and $r^g \geq 2c$. If $r^g < 2c$ then $(R_F, 1-R_F)$ is never the equilibrium profile. That is, $(R_G, 0)$ is the equilibrium profile - since $R_G$ does not depend on $m$ then any $m \in [-r^f, r^f]$ is optimal.

    \item Case 2: $\lambda < 1$: For the profile $(R_F, 1-R_F)$ to be the equilibrium and to maximize the objective function, we need to choose the maximal $m$ which satisfies our conditions. 
    
    The maximal $m = m_{max} = m_-$ is optimal only if it holds that $R_F(1-\lambda) + \lambda \geq R_G$. Plugging $m = m_{max}$ results in:
    \begin{align*}
        \frac{r^f + m_{max}}{2r^f}(1-\lambda) + \lambda - R_G &= \frac{r^f + r^f(4R_G - 3)}{2r^f}(1-\lambda) + \lambda - R_G \\
        &= (2R_G - 1)(1-\lambda) + \lambda - R_G \\
        &= (R_G - 1)(1-2\lambda)
    \end{align*}

    Therefore, we can further split the analysis of $\lambda$ into 3 cases:
    \begin{itemize}
        \item If $\lambda > 0.5$ and $r^g \geq 2c$ then $m = m_{max}$ is optimal.
        \item If $\lambda < 0.5$ then the profile $(R_G, 0)$ maximizes our objective function. Therefore, if $r^g \geq 2c$ then any $m \in (m_-, r^f]$ is optimal, and if $r^g < 2c$ then any $m \in [-r^f, r^f]$ is optimal.
        \item If $\lambda = 0.5$ then both profiles can maximize our objective function. If $r^g < 2c$ then any $m \in [-r^f, r^f]$ is optimal and if $r^g \geq 2c$ then any $m \in [m_-, r^f]$ is optimal.
    \end{itemize}   
\end{enumerate}

Therefore, to summarize, if $r^g < 2c$ then any $m \in [-r^f, r^f]$ is optimal.
If $r^g \geq 2c$ then 
\begin{align*}
    m = \begin{cases}
        -r^f & \mbox{$\lambda \geq 1$} \\
        r^f(4R_G - 3) & \mbox{$0.5 < \lambda < 1$} \\
        m' & \mbox{$\lambda = 0.5$} \\
        m'' & \mbox{$\lambda < 0.5$}
    \end{cases}
\end{align*}
where $m', m''$ can be any value such that $m' \in [r^f(4R_G - 3), r^f]$ and $m'' \in (r^f(4R_G - 3), r^f]$.

This concludes the proof of \Cref{cor:welfare}.
\end{proof}

\end{document}